\documentclass[aps,pra,twocolumn,superscriptaddress,floatfix,nofootinbib,showpacs,longbibliography]{revtex4-1}

\usepackage[utf8]{inputenc}  
\usepackage[T1]{fontenc}     
\usepackage[british]{babel}  
\usepackage[sc,osf]{mathpazo}
\usepackage[scaled=0.86]{berasans}  
\usepackage[colorlinks=true, citecolor=blue, urlcolor=blue]{hyperref}
\makeatletter
\newcommand{\setword}[2]{%
  \phantomsection
  #1\def\@currentlabel{\unexpanded{#1}}\label{#2}%
}
\makeatother
\usepackage{graphicx} 
\usepackage[babel]{microtype}  
\usepackage{amsmath,amssymb,amsthm,bm,amsfonts,mathrsfs,bbm} 

\usepackage{xspace}  
\usepackage{pgf,tikz}
\usepackage{xcolor}
\usepackage{multirow}
\usepackage{array}
\usepackage{bigstrut}
\usepackage{braket}
\usepackage{color}
\usepackage{natbib}
\usepackage{multirow}
\usepackage{mathtools}
\usepackage{float}
\usepackage[caption = false]{subfig}
\usepackage{xcolor,colortbl}
\usepackage{color}

\newcommand{\be}{\begin{equation}}
\newcommand{\ee}{\end{equation}}
\newcommand{\ba}{\begin{eqnarray}}
\newcommand{\ea}{\end{eqnarray}}
\newcommand{\ketbra}[2]{|#1\rangle \langle #2|}

\newtheorem{theorem}{Theorem}
\newtheorem{corollary}{Corollary}
\newtheorem{definition}{Definition}
\newtheorem{proposition}{Proposition}

\newtheorem{lemma}{Lemma}

\def\>{\rangle}
\def\<{\langle}

\usepackage{centernot}
\usepackage{subfig}

\begin{document}

\title{Local Quantum State Marking}

\author{Samrat Sen}
\affiliation{School of Physics, IISER Thiruvananthapuram, Vithura, Kerala 695551, India.}

\author{Edwin Peter Lobo}
\affiliation{School of Physics, IISER Thiruvananthapuram, Vithura, Kerala 695551, India.}

\author{Sahil Gopalkrishna Naik}
\affiliation{School of Physics, IISER Thiruvananthapuram, Vithura, Kerala 695551, India.}

\author{Ram Krishna Patra}
\affiliation{School of Physics, IISER Thiruvananthapuram, Vithura, Kerala 695551, India.}

\author{Tathagata Gupta}
\affiliation{Physics and Applied Mathematics Unit, Indian Statistical Institute, 203 B.T. Road, Kolkata 700108, India.}

\author{Subhendu B. Ghosh}
\affiliation{Physics and Applied Mathematics Unit, Indian Statistical Institute, 203 B.T. Road, Kolkata 700108, India.}

\author{Sutapa Saha}
\affiliation{Physics and Applied Mathematics Unit, Indian Statistical Institute, 203 B.T. Road, Kolkata 700108, India.}

\author{Mir Alimuddin}
\affiliation{School of Physics, IISER Thiruvananthapuram, Vithura, Kerala 695551, India.}

\author{Tamal Guha}
\affiliation{Department of Computer Science, The University of Hong Kong, Pokfulam road 999077, Hong Kong.}

\author{Some Sankar Bhattacharya}
\affiliation{International Centre for Theory of Quantum Technologies (ICTQT), University of Gdask, Bazynskiego 8, 80-309 Gdansk, Poland.}

\author{Manik Banik}
\affiliation{School of Physics, IISER Thiruvananthapuram, Vithura, Kerala 695551, India.}

\begin{abstract}
We propose the task of local state marking (LSM), where some multipartite quantum states chosen randomly from a known set of states are distributed among spatially separated parties without revealing the identities of the individual states. The collaborative aim of the parties is to correctly mark the identities of states under the restriction that they can perform only local quantum operations (LO) on their respective subsystems and can communicate with each other classically (CC) -- popularly known as the operational paradigm of LOCC. While mutually orthogonal states can always be marked exactly under global operations, this is in general not the case under LOCC. We show that the LSM task is distinct from the vastly explored task of local state distinguishability (LSD) -- perfect LSD always implies perfect LSM, whereas we establish that the converse does not hold in general. We also explore entanglement assisted marking of states that are otherwise locally unmarkable and report intriguing entanglement assisted catalytic LSM phenomenon.
\end{abstract}


\maketitle	
\section{Introduction} 
Discrimination task, wherein the aim is to distinguish among physical or mathematical objects {\it viz.} states, processes, circuits, probability distributions, is one of the rudimentary steps that appear in information processing protocols, statistical inference, and hypothesis testing \cite{Shannon48,Lehmann05}. Distinct objects or perfectly distinguishable states of a system can be used to store information which assures readability of the information without any ambiguity. Information protocols in the quantum world \cite{Wiesner83,Bennett84,Ekert91,Bennett92,Bennett92(1),Bennett93}, however, are governed by rules that are fundamentally different from our classical worldview. For instance, classical information encoded in non-orthogonal quantum states, either pure or mixed, cannot be perfectly decoded since the no-cloning theorem \cite{Wootters82} ( more generally the no-broadcasting theorem \cite{Barnum96}) puts restriction on their perfect discrimination. Such a constraint is strictly quantum (more precisely, non-classical \cite{Barnum07,Banik19}) in nature as pure classical states are always perfectly distinguishable \cite{Self1}. While a set of mutually orthogonal quantum states can always be distinguished perfectly, interesting situations arise for multipartite quantum systems when discriminating operations among the spatially separated parties holding different subsystems are limited to local quantum operation assisted with classical communication (LOCC). This constitutes the framework for the problem of local state discrimination (LSD) \cite{Bennett99,Walgate00,Ghosh01,Walgate02,Ghosh04,Horodecki03,Watrous05,Hayashi06}. During the last two decades LSD has been studied in great detail resulting in a plethora of interesting conclusions \cite{Bennett99(1),DiVincenzo03,Niset06,Duan07,Calsamiglia10,Bandyopadhyay11,Chitambar14,Halder18,Demianowicz18,Halder19,Halder19(1),Agrawal19,Rout19,Bhattacharya20,Banik20,Rout20} and it also finds applications in useful tasks \cite{Terhal01,DiVincenzo02,Eggeling02,Markham08,Matthews09}. Apart from LSD and more general quantum state discrimination problems \cite{Helstrom69,Holevo73,Yuen75}, several other discrimination tasks, {\it eg.} channel/sub-channel discrimination, process discrimination, circuit discrimination, have been studied during the recent past \cite{Chiribella08,Piani09,Chiribella12,Hirche21} that subsequently motivate several novel information protocols \cite{Pirandola19,Takagi19,Takagi19(1),Chiribella21,Bhattacharya21}.   
In this paper, we introduce a novel variant of discrimination task, which we call local state marking (LSM). A subset of states chosen randomly from a known set of multipartite states is provided to spatially separated parties without revealing the identities of the individual states. The aim is to mark the identities of the states under the operational paradigm of LOCC. For a given set of multipartite states $\mathcal{S}$ one can, in fact, define a class of discrimination tasks denoted by $m$-LSM. Here $1\le m\le |\mathcal{S}|$ with $1$-LSM corresponding to the task of LSD and  the $|\mathcal{S}|$-LSM task we will denote simply as LSM where $|\mathcal{S}|$ is the cardinality of set $\mathcal{S}$. It turns out that the task of LSM is distinct from the task of LSD. In particular, we show that local distinguishability of an arbitrary set of states always implies local markability, but the converse does not always hold true. We provide an example of mutually orthogonal states that are locally markable but not locally distinguishable. We then provide examples of orthogonal states that can neither be distinguished nor marked perfectly under local operations. Some generic implications between $m$-LSM and $m^\prime$-LSM tasks are also analyzed when $m\neq m^\prime$. We then study entanglement assisted local marking of states where additional entanglement is provided as a resource to mark the states that are otherwise locally unmarkable. There we report intriguing entanglement assisted catalytic LSM phenomenon-- a locally unmarkable set of states can be perfectly marked when additional entanglement is supplied as a resource. Interestingly, the entanglement is returned (either partially or completely) once the marking task is done.

\section{Notations and preliminaries} 
A quantum system prepared in a pure state is represented by a vector $\ket{\psi}\in\mathcal{H}$, where $\mathcal{H}$ is the Hilbert space associated with the system. Throughout this work we will consider finite dimensional quantum systems and consequently $\mathcal{H}$ will be isomorphic to some complex Euclidean space $\mathbb{C}^d$. An $N$ partite quantum system is associated with the tensor product Hilbert space $\bigotimes_{i=1}^N\mathbb{C}^{d_i}_{A_i}$, where $\mathbb{C}^{d_i}_{A_i}$ is the Hilbert space of the $i^{th}$ subsystem held by the $i^{th}$ party \cite{Self2,Carcassi21}. A state $\ket{\psi}_{A_1\cdots A_N}\in\bigotimes_{i=1}^N\mathbb{C}^{d_i}_{A_i}$ is called separable across $\mathcal{A}$ {\it vs} $\mathcal{A}^\mathsf{C}$ cut if it is of the form $\ket{\psi}_{A_1\cdots A_N}=\ket{\psi}_{\mathcal{A}}\otimes\ket{\psi}_{\mathcal{A}^\mathsf{C}}$, where $\ket{\psi}_{\mathcal{A}}\in\bigotimes_{i|A_i\in\mathcal{A}}\mathbb{C}^{d_i}_{A_i}$ and $\ket{\psi}_{\mathcal{A}^\mathsf{C}}\in\bigotimes_{i|A_i\in\mathcal{A}^\mathsf{C}}\mathbb{C}^{d_i}_{A_i}$ with $\mathcal{A}$ being a proper nonempty subset of $\mathbb{A}\equiv\{A_1,\cdots,A_N\}$ and $\mathcal{A}^\mathsf{C}\equiv\mathbb{A}\setminus\mathcal{A}$. A multiparty state $\ket{\psi}_{A_1\cdots A_N}$ is fully separable if it is separable across all possible bipartite cuts, {\it i.e.,} $\ket{\psi}_{A_1\cdots A_N}=\bigotimes_{i=1}^N\ket{\psi}_{A_i}$ with $\ket{\psi}_{A_i}\in\mathbb{C}^{d_i}_{A_i}$. For the sake of notational brevity, we will avoid the party index when there is no confusion. A set of quantum states is perfectly distinguishable whenever they are pairwise orthogonal. Moreover, in accordance with the no-cloning theorem \cite{Wootters82}, pairwise orthogonality turns out to be the necessary requirement for perfect distinguishability. 

In the multipartite scenario, when different parts of the quantum systems are held by spatially separated parties, the class of operations LOCC captures the `distant lab' paradigm. Although it is extremely hard to characterize structure of LOCC operations \cite{Chitambar14(1)}, this restricted paradigm plays crucial role to understand the resource of quantum entanglement and it constitutes the scenario for the task of $m$-LSM. 
\begin{definition}\label{def1}
[$m$-LSM] $m$ number of states chosen randomly from a known set of pairwise orthogonal N-party quantum states $\mathcal{S}\equiv\left\{\ket{\psi_j}~|~\langle\psi_i|\psi_j\rangle=\delta_{ij}\right\}\subset\bigotimes_{i=1}^N\mathbb{C}^{d_i}_{A_i}$ are distributed among spatially separated parties without revealing the identity of each state. The $m$-LSM task is to perfectly identify/mark each of the states under the operational paradigm of LOCC.
\end{definition}
In Definition \ref{def1}, $m$ can take values from $1$ to $|\mathcal{S}|$ and accordingly they constitute different discrimination tasks (see Fig.\ref{fig0}). The task of $1$-LSM is more popular as LSD which has been explored in great detail during the last two decades \cite{Bennett99,Walgate00,Ghosh01,Walgate02,Ghosh04,Horodecki03,Watrous05,Hayashi06,Bennett99(1),DiVincenzo03,Niset06,Duan07,Calsamiglia10,Bandyopadhyay11,Chitambar14,Halder18,Demianowicz18,Halder19,Halder19(1),Agrawal19,Rout19,Bhattacharya20,Banik20,Rout20}. The problem of LSD has also been studied with ensembles containing non-orthogonal states \cite{Peres91,Chitambar13}. Similarly, Definition \ref{def1} can also be generalized for such ensembles. In that case the quantity of interest will be the difference between maximum success probabilities of the corresponding marking task under global and local operations, respectively. 
\begin{figure}[t!]
\centering
\includegraphics[width=0.48\textwidth]{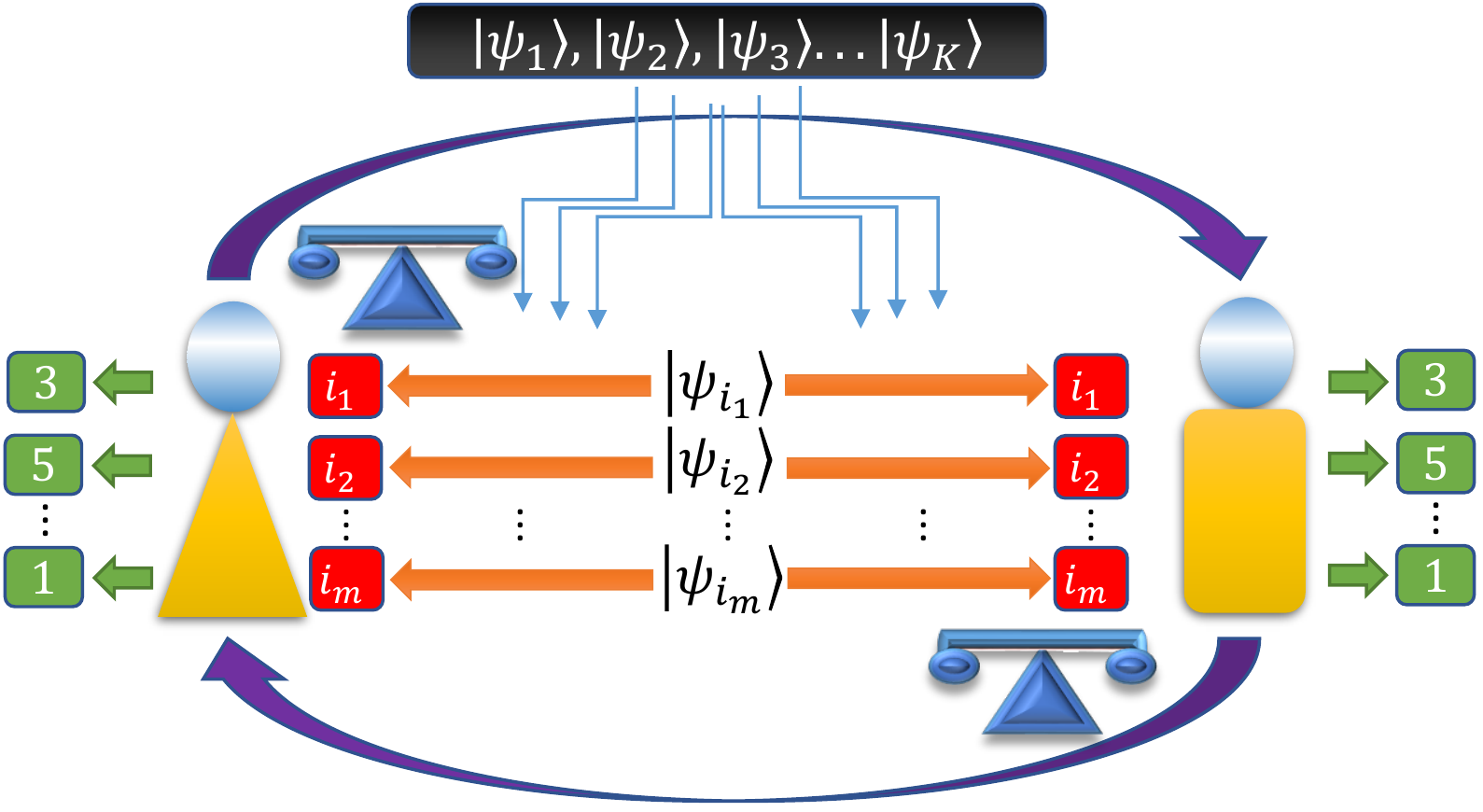}
\caption{(Color online) The task of $m$-LSM is illustrated for the bipartite scenario. $m$ states chosen randomly from a set of $K$ states are distributed between spatially separated Alice and Bob without revealing the identities of the individual states. They have to identify the indices $i_1,\cdots,i_m$ using LOCC. In this particular example the indices are identified to be $(i_1=3,i_2=5,\cdots,i_m=1$). The special case of $m=1$ corresponds to the task of LSD.}\label{fig0}
\end{figure} 
  
\section{Results} 
We will start the technical part of this article by establishing some generic results.
\begin{lemma}\label{lemma1}
For a set of multipartite states $\mathcal{S}$, perfect $(|\mathcal{S}|-2)$-LSM always implies perfect LSM.
\end{lemma}
\begin{proof}
Perfect $(|\mathcal{S}|-2)$-LSM of the set $\mathcal{S}$ implies that given arbitrary $(|\mathcal{S}|-2)$ states from the set, they can be marked locally. So we are left with two more states to identify locally. According to a standard result by Walgate {\it et al.}, any two multipartite pure orthogonal states can be distinguished locally \cite{Walgate00}, which proves our claim. 
\end{proof}
While proof of Lemma \ref{lemma1} follows straightforwardly from the result of Walgate {\it et al.}, in the next we establish a rather nontrivial thesis.  
\begin{theorem}\label{theo1}
For a set of multipartite states $\mathcal{S}$, perfect LSD (i.e. $1$-LSM) always implies perfect LSM (i.e. $|\mathcal{S}|$-LSM).
\end{theorem}
\begin{proof}
Let the set of states $\mathcal{S}_K\equiv\left\{\ket{\psi_1},\cdots,\ket{\psi_K}\right\}\subset\bigotimes_{i=1}^N\mathbb{C}^{d_i}_{A_i}:=\mathcal{H}$ be locally distinguishable. The problem of LSM for the set $\mathcal{S}_K$ can be reformulated as an LSD problem of the set of states $\mathcal{S}_{\mathcal{P}[\{K\}]}\equiv\left\{\mathcal{P}\left(\otimes_{i=1}^K\ket{\psi_i}\right)\right\}\subset\mathcal{H}^{\otimes K}$, where $\left\{\mathcal{P}\left(\otimes_{i=1}^K\ket{\psi_i}\right)\right\}$ denotes the set of tensor product states generated through permutations of the indices $\{1,\cdots,K\}$. For instance, $\mathcal{S}_{\mathcal{P}[\{3\}]}:=\left\{\mathcal{P}\left(\otimes_{i=1}^3\ket{\psi_i}\right)\right\}\equiv\{\ket{\psi_1\psi_2\psi_3},$ $\ket{\psi_1\psi_3\psi_2},~\ket{\psi_2\psi_3\psi_1},~\ket{\psi_2\psi_1\psi_3},~\ket{\psi_3\psi_2\psi_1},~\ket{\psi_3\psi_1\psi_2}\}$, where $\ket{x~y~z}:=\ket{x}\otimes\ket{y}\otimes\ket{z}$. The states in $\mathcal{S}_{\mathcal{P}[\{K\}]}$ can be expressed group-wise as follows,
\begin{align*}
\mathcal{G}_l:=\ket{\psi_l}\otimes\mathcal{S}_{\mathcal{P}[\{K\}\setminus l]}\equiv\ket{\psi_l}\otimes\left\{\mathcal{P}\left(\otimes_{i\neq l}\ket{\psi_i}\right)\right\},
\end{align*}
where $l\in\{1,\cdots,K\}$. Clearly, the groups $\mathcal{G}_l$ make disjoint partitions of the set $\mathcal{S}_{\mathcal{P}[\{K\}]}$, {\it i.e.,} $\mathcal{S}_{\mathcal{P}[\{K\}]}\equiv \bigcup_{l=1}^K \mathcal{G}_l$ s.t. $\mathcal{G}_l\cap\mathcal{G}_{l^\prime}=\emptyset$ whenever $l\neq l^\prime$. Since the states in $\mathcal{S}_K$ are locally distinguishable, by local operations on the first part of the tensor product states in $\mathcal{S}_{\mathcal{P}[\{K\}]}$ we can know with certainty in which of the above groups the given state lies. If the group turns out to be $\mathcal{G}_{l^\star}$ (i.e., if the index $l$ has been identified to be $l^*$), the given state $\ket{\psi_{l^\star}}\otimes(\cdots)$ evolves to $\ket{\psi^\prime_{l^\star}}\otimes(\cdots)$ due to the LOCC protocol, where the term within the brackets remain unchanged and hence further LOCC protocols can be applied on them. The group of states $\mathcal{G}_{l^\star}=\ket{\psi^\prime_{l^\star}}\otimes\mathcal{S}_{\mathcal{P}[\{K\}\setminus {l^\star}]}$ can be further partitioned into disjoint subsets as,
\begin{align*}
\mathcal{G}_{l^\star}\equiv\bigcup\mathcal{G}_{l^\star,m}~~\mbox{s.t.}~~\mathcal{G}_{l^\star,m}\cap\mathcal{G}_{l^\star,m^\prime}=\emptyset~\forall~m\neq m^\prime\nonumber,\\
\mbox{where}~\mathcal{G}_{l^\star,m}:=\ket{\psi^\prime_{l^\star}}\otimes\ket{\psi_{m}}\otimes\mathcal{S}_{\mathcal{P}[\{K\}\setminus \{l^\star,m\}]},
\end{align*} 
and $m,m^\prime\in\{1,\cdots,K\}\setminus l^\star$. Since any subset of a locally distinguishable set of states is also locally distinguishable, the identity of the index $m$ can be known perfectly by applying some local protocol on the $\ket{\psi_{m}}$ part of the given state. As before, the remaining parts of the state will not change. We can continue this process till we completely determine the identity of the state in $\mathcal{S}_{\mathcal{P}[\{K\}]}$ which in turn marks the state in $\mathcal{S}_{K}$. This completes the proof.  \end{proof}
While Theorem \ref{theo1} deals with the implications between two extreme cases, particularly establishing $1$-LSM $\implies|\mathcal{S}|$-LSM, the following corollaries establish few more nontrivial implications among generic $m$-LSM tasks.  
\begin{corollary}\label{coro1}
For a set of multipartite states $\mathcal{S}$, perfect $m$-LSM always implies perfect $m^\prime$-LSM, where $1\le m\le m^\prime(:=nm)\le|\mathcal{S}|$ with $n\in\mathbb{N}$.
\end{corollary}
\begin{proof}
Given a set $\mathcal{S}_K\equiv\left\{\ket{\psi_1},\cdots,\ket{\psi_K}\right\}\subset\bigotimes_{i=1}^N\mathbb{C}^{d_i}_{A_i}$ is $m$-LSM we are supposed to prove that it is $m^\prime$-LSM, where $m^{\prime} = n m$ with $n\in \mathbb{N}$. Intuitively, the proof goes as follows. Let $\mathcal{L}_m$ be the local protocol that successfully completes the $m$-LSM task for the set $\mathcal{S}_K$. For the $m^\prime$-LSM task, we divide the set of $m^{\prime}$ states into $n$ arbitrary disjoint sets each containing $m$ states. Treating each of these $n$ sets independently, we can mark them locally by following the protocol $\mathcal{L}_m$. Thus, by successively applying the protocol $\mathcal{L}_m$ we can construct the local protocol $\mathcal{L}_{m^\prime}$ for the $m^\prime$-LSM task .
	
We can also reformulate this as an LSD task as was done in Theorem ${\bf 1}$. We begin by noting that from the set $\mathcal{S}_K\equiv\left\{\ket{\psi_1},\cdots,\ket{\psi_K}\right\}\subset\bigotimes_{i=1}^N\mathbb{C}^{d_i}_{A_i}$ one can choose $m$ states in $\prescript{K}{}{C}_{m}$ different ways. Let denote each such choice of states by the set $\mathcal{S}_m^j$ where $j \in \{1,\cdots,\prescript{K}{}{C}_{m}\}$. 	Therefore, $m$-LSM problem of $\mathcal{S}_K$ can be reformulated as the LSD problem of the set of states
\begin{align*}\label{m-lsm}
\mathcal{S}_{\prescript{K}{}{C}_{m}\times m!} &\equiv\bigcup_{j =1}^{\prescript{K}{}{C}_{m}}\mathcal{S}_{\mathcal{P}[\{m\}]}^{j}\\
\mbox{s.t.}~~\mathcal{S}_{\mathcal{P}[\{m\}]}^{j}&\bigcap\mathcal{S}_{\mathcal{P}[\{m\}]}^{j^\prime}=\emptyset~\mbox{for}~j\neq j^\prime,
\end{align*}
where $\mathcal{S}_{\mathcal{P}[\{m\}]}^{j}$ is defined similarly as in Theorem ${\bf 1}$. We are given that perfect $m$-LSM of the the set $\mathcal{S}_K$ is possible, {\it i.e.,} there exists a local protocol $\mathcal{L}_m$ that perfectly distinguishes the states in $\mathcal{S}_{\prescript{K}{}{C}_{m}\times m!}$. While considering the $m^\prime$-LSM problem, or equivalently, the LSD problem of the set $\mathcal{S}_{\prescript{K}{}{C}_{m^\prime}\times m^\prime!}$, the states in $\mathcal{S}_{\mathcal{P}[\{m^\prime\}]}^{j}$ can be expressed group-wise as 	$\mathcal{G}^{j}_{l_1,...,l_{m}}:=\ket{\psi_{l_1},...,\psi_{l_{m}}}\otimes\mathcal{S}_{\mathcal{P}[\{m^\prime\}\setminus {\{l_1,...,l_{m}\}}]}^{j}$ for each value of $j$. Thus the groups $\mathcal{G}^{j}_{l_1,...,l_{m}}$ make a disjoint partition of the the set $\mathcal{S}_{\prescript{K}{}{C}_{m^\prime}\times m^\prime!}$. Since $\mathcal{S}_K$ is $m$-LSM, by  performing local operations on the first $m$-parts of the tensor product states in $\mathcal{S}_{\prescript{K}{}{C}_{m^\prime}\times m^\prime!}$ we can fix the indices $l_1,..,l_{m}$ of $\mathcal{G}^{j}_{l_1,...,l_{m}}$. If $l_1,..,l_{m}$ is identified to be $l^*_1,..,l^*_{m}$ then we know with certainty that the given state lies in $\bigcup\limits_{j}\mathcal{G}^{j}_{l^*_1,...,l^*_{m}}$ and the given state is identified to be of the form $\ket{\psi_{l^*_1},...,\psi_{l^*_{m}}}\otimes(...)$ and evolves to $\ket{\psi'_{l^*_1},...,\psi'_{l^*_{m}}}\otimes(...)$ after the protocol has been performed, where the terms in the brackets remain unchanged and hence further protocols can be performed on that part. The groups $\mathcal{G}^{j}_{l^*_1,...,l^*_{m}} = \ket{\psi'_{l^*_1},...,\psi'_{l^*_{m}}}\otimes\mathcal{S}_{\mathcal{P}[\{m^\prime\}\setminus {\{l_1,...,l_{m}}\}]}^{j}$ can be further partitioned into disjoint subsets as $\mathcal{G}^{j}_{l^*_1,...,l^*_{m}}=\bigcup\limits_{t_1,\cdots,t_{m}}\mathcal{G}^{j}_{l^*_1,...,l^*_{m},t_1,...,t_{m}}$ for each value of $j$,
where $ \mathcal{G}^{j}_{l^*_1,...,l^*_{m},t_1,...,t_{m}} \equiv \ket{\psi'_{l^*_1},...,\psi'_{l^*_{m}}}\otimes\ket{\psi_{t_1},...,\psi_{t_{m}}}\otimes\mathcal{S}_{\mathcal{P}[\{m^\prime\}\setminus {\{l^*_1,...,l^*_{m},t_1,...,t_{m}}\}]}^{j}$. Since $\mathcal{S}_K$ is $m$-LSM, any subset of $\mathcal{S}_K$ is also $m$-LSM. Hence we can further fix the indices $t_1,..,t_{m}$ by performing local operations on the second $m$-parts of the tensor product states in $\bigcup\limits_j\mathcal{G}^{j}_{l^*_1,...,l^*_{m}}$. We continue this process $n$-times which will completely fix all the $nm=m^\prime$ indices. This will also fix the index $j$ since we have completely distinguished the state. This completes the proof. For the special case of $m=1$, LSD (i.e., $1$-LSM) implies perfect $m$-LSM with $1\le m\le|\mathcal{S}|$.
\end{proof}
\begin{corollary}\label{coro2}
For a set of multipartite states $\mathcal{S}$ containing only product states, perfect $m$-LSM always implies perfect $m'$-LSM, where $1\le m\le m^\prime\le|\mathcal{S}|$.
\end{corollary}
\begin{proof}
It is sufficient to show that $m$-LSM implies $(m+1)$-LSM for any set $\mathcal{S}$ containing product states only. Given $(m+1)$-states to be marked, we begin by marking the first $m$-states by using the protocol for $m$-LSM. However, during this process the first $m$-states are destroyed. But since we have determined the identity of the first $m$-states, we can locally create the original set of first $m$-states once again. It is to be noted that we can locally create any set of multipartite states whose identity is known if and only if the set contains product states only. We now run the protocol for $m$-LSM once again but this time on the last $m$-states. Thus we have identified all the $(m+1)$-states. This completes the proof. Reformulation of this proof in terms of the LSD problem is straightforward and follows in a similar fashion as the proof of Corollary ${\bf 1}$.
\end{proof}
We note in passing that $m$-LSM does not trivially imply $(m-1)$-LSM for product states. In the $m$-LSM task, $m$ states are accessible to the parties in order to mark their identities. However, for $(m-1)$-LSM, the number of accessible states reduce to $(m-1)$ and no trivial inferences can be drawn. Although we were unable to prove for product states, whether m-LSM implies (m-1)-LSM or not, if this really were the case, then the consequences would be very exciting. Using the contrapositive of the statement that m-LSM implies (m-1)-LSM, we would have thus given a protocol for constructing locally indistinguishable states in higher dimensions starting from states like Bennett's UPBs for which perfect LSD is not possible.

Our next result, however, establishes that the converse statement of Theorem \ref{theo1} does not hold in general and we show this by providing an explicit example wherein given a set of entangled states, known to be locally indistinguishable, the task of LSM is still  possible, and that too with a substantial amount of surplus entanglement shared between the distant parties at the end of the protocol.
\begin{theorem}\label{theo2}
Perfect LSM of a given set of states $\mathcal{S}$ does not necessarily imply perfect LSD of $\mathcal{S}$.
\end{theorem}
\begin{proof}
The proof is constructive. We provide a set of pairwise orthogonal states that can be perfectly marked under LOCC but does not allow perfect local distinguishability. To this aim consider the set of states $\mathcal{X}_4\equiv\{\ket{\chi_i}\}_{i=1}^4\subset\mathbb{C}^4_A\otimes\mathbb{C}^4_B$ shared between Alice and Bob, where  
\small
\begin{align*}
\ket{\chi_1}&:=\ket{\phi^+}_{A_1B_1}\otimes\ket{\phi^+}_{A_2B_2},~\ket{\chi_2}:=\ket{\phi^-}_{A_1B_1}\otimes\ket{\phi^-}_{A_2B_2},\\
\ket{\chi_3}&:=\ket{\psi^+}_{A_1B_1}\otimes\ket{\phi^-}_{A_2B_2},~\ket{\chi_4}:=\ket{\psi^-}_{A_1B_1}\otimes\ket{\phi^-}_{A_2B_2},\\
\mbox{with}&~~\ket{\phi^\pm}:=\frac{\ket{00}\pm\ket{11}}{\sqrt{2}},~~~\ket{\psi^\pm}:=\frac{\ket{01}\pm\ket{10}}{\sqrt{2}},
\end{align*}
\normalsize
and $A_1,A_2$ subsystems are with Alice while $B_1,B_2$ are with Bob. The part of $\ket{\chi_i}$ indexed with $A_1B_1$ we will call the first part and the part with index $A_2B_2$ will be the second part. For LSM, Alice and Bob are provided the state $\ket{\chi_p}\otimes\ket{\chi_q}\otimes\ket{\chi_r}\otimes\ket{\chi_s}\in\left(\mathbb{C}^4_A\otimes\mathbb{C}^4_B\right)^{\otimes4}$, without specifying the indices $p,q,r,s\in\{1,\cdots,4\}$ and $p, q, r, s$ are all distinct. Their collaborative aim is to identify the indices where the collaboration is restricted to LOCC. It turns out that there exists a local strategy that marks the states exactly (detailed protocol is provided in Appendix \ref{appendix}). Now, local indistinguishability of the set $\mathcal{X}_4$ follows from the results of Yu {\it et al} \cite{Yu12}. There the authors have proved that states in $\mathcal{X}_4$ cannot be distinguished perfectly by any PPT POVM, a larger class of operations that strictly contains all LOCC operations. This completes the proof.
\end{proof}
It turns out that at the end of local marking strategy described in Theorem \ref{theo2} $3$-ebit of entanglement (on average) remains between Alice and Bob (see Appendix \ref{appendix}). However, the optimality of the protocol in terms of retaining entanglement remains an open problem. 

So far we have considered sets that are locally markable. Our next result provides an example of a set of mutually orthogonal states that cannot be marked perfectly under LOCC.
\begin{proposition}\label{prop1}
The two qubit Bell basis $\mathcal{B}_4\equiv\{\ket{b_1}:=\ket{\phi^+},\ket{b_2}:=\ket{\phi^-},\ket{b_3}:=\ket{\psi^+},\ket{b_4}:=\ket{\psi^-}\}\subset\mathbb{C}^2\otimes\mathbb{C}^2$ is locally unmarkable. 
\end{proposition}
\begin{proof}
LSM of $\mathcal{B}_4$ is equivalent to  LSD of the set $\mathcal{B}_{\mathcal{P}[\{4\}]}$ that contains $24$ pairwise orthogonal maximally entangled states in $\mathbb{C}^{16}\otimes\mathbb{C}^{16}$. Therefore, the desired thesis follows from the fact that $n$ pairwise orthogonal maximally entangled states in $\mathbb{C}^{d}\otimes\mathbb{C}^{d}$ cannot be perfectly distinguished locally  whenever $n>d$ \cite{Hayashi06}. 
\end{proof}
Furthermore, in accordance with Lemma \ref{lemma1}, the set $\mathcal{B}_4$ does not allow perfect $2$-LSM and it also straightforwardly follows that perfect $3$-LSM of $\mathcal{B}_4$ is impossible (in fact, for any set $\mathcal{S}$, $(|\mathcal{S}|-1)$-LSM always implies $|\mathcal{S}|$-LSM). A generalization of Proposition \ref{prop1} follows arguably. 
\begin{proposition}\label{prop2}
Consider any set of maximally entangled states $\mathcal{B}_K(d):=\left\{\ket{b_i}~|~\langle b_i|b_j\rangle=\delta_{ij}\right\}_{i=1}^K\subset\mathbb{C}^d\otimes\mathbb{C}^d$. The set is locally unmarkable whenever $K!>d^K$. 	
\end{proposition}  
We now move on to the possibility of entanglement assisted marking of states that otherwise are locally unmarkable.
It might happen that given $\delta$-ebit of entanglement some LSM task can be performed exactly which is otherwise impossible to do locally, and moreover $\epsilon$-ebit of entanglement is left at the end of the protocol. Such a protocol we will call $(\delta,\epsilon)$ entanglement catalytic protocol and $(\delta-\epsilon)$ quantifies the amount of entanglement consumed to accomplish the given LSM task.

Recall that given $1$-ebit of entanglement as additional resource, the two-qubit Bell basis can be distinguished perfectly. One of the party teleports his/her part of the unknown Bell state to the other party who then performs the Bell basis measurement to identify the state. Furthermore, it is known that $1$-ebit entanglement is the necessary resource required for perfect discrimination of the $2$-qubit Bell basis \cite{Ghosh01}. Coming to the question of entanglement assisted marking of the set $\mathcal{B}_4$, we obtain the following result. 
\begin{proposition}\label{prop3}
There exists a $(2,1)$ entanglement catalytic perfect protocol for LSM of the set $\mathcal{B}_4$. 
\end{proposition}  
\begin{proof}
LSM of the set $\mathcal{B}_4$ is equivalent to LSD of the set $\mathcal{B}_{\mathcal{P}[\{4\}]}$ containing states of the form $\ket{b_p}\otimes\ket{b_q}\otimes\ket{b_r}\otimes\ket{b_s}$ with $p,q,r,s\in\{1,\cdots,4\}~\&~p, q, r, s$ are distinct. Let some supplier provides two EPR states for discriminating the set $\mathcal{B}_{\mathcal{P}[\{4\}]}$. Using the teleportation protocol Alice and Bob can know any of the two indices among $p,q,r,s$. Say they identify the indices $p$ and $q$. Then the value of $r$ has only two possibilities and the result of Walgate {\it et al.} ensures that this value can be known exactly under LOCC \cite{Walgate00}. While determining the value of $r$, the entanglement of the state $\ket{b_r}$ gets destroyed. However, at the end of the protocol, entanglement of the state $\ket{b_s}$ remains intact and its identity is also known. Therefore, $1$-ebit of entanglement can be returned back to the supplier. So the protocol consumes $1$-ebit of entanglement in catalytic sense.
\end{proof}
Once again we are not sure about optimality of the protocol in Proposition \ref{prop3} in terms of resource consumption, and leave the question open here for further research. One can obtain a more exotic example of entanglement catalytic local marking phenomenon. To this aim we first prove the following result. 
\begin{proposition}\label{prop4}
Any three Bell states of the two-qubit system is unmarkable under one-way LOCC. 
\end{proposition} 
\begin{proof}
Consider a set of maximally entangled states $\{\ket{\phi_{k}}\}_{k=1}^{n}\subset \mathbb{C}^{d}\otimes\mathbb{C}^{d}$ that can be obtained from $\ket{\phi_{0}(d)}:=\frac{1}{\sqrt{d}}\sum_{i=0}^{d-1}\ket{ii}$ by applying some unitary on one part of the system, {\it i.e.,} $\ket{\phi_{k}}=(U_{k}\otimes\mathbb{I})\ket{\phi_{0}(d)}$. According to a  criterion, conjectured in \cite{Ghosh04} and subsequently derived in \cite{Bandyopadhyay11(0)}, the states can be discriminated under one-way LOCC \textit{if and only if} there exists a $\ket{\psi}\in\mathbb{C}^{d}$, such that $\bra{\psi_{i}}\psi_{j}\rangle=\delta_{ij}, \forall i,j\in\{1,2,\cdots, n\}$, where $\ket{\psi_{k}}=U_{k}\ket{\psi}$. In our case, without loss of any generality we can consider $\mathcal{B}_3\equiv\{\ket{b_{1}}:=\ket{\phi^{+}}, \ket{b_{2}}:=\ket{\phi^{-}},\ket{b_{3}}=\ket{\psi^{+}}\}$, such that  $\mathcal{B}_{\mathcal{P}[\{3\}]}\equiv\{\ket{\phi_{k}}\}_{k=1}^{6}\subset \mathbb{C}^{8}\otimes\mathbb{C}^{8}$, where
\footnotesize
\begin{align*}
\ket{b_{1}b_{2}b_{3}}:=\ket{\phi_{1}}=(U_1\otimes\mathbb{I}_{8})\ket{\phi_{0}(8)}=([\mathbb{I}_{2}\otimes\sigma_{z}\otimes\sigma_{x}]\otimes\mathbb{I}_{8})\ket{\phi_{0}(8)},\\
\ket{b_{1}b_{3}b_{2}}:=\ket{\phi_{2}}=(U_2\otimes\mathbb{I}_{8})\ket{\phi_{0}(8)}=([\mathbb{I}_{2}\otimes\sigma_{x}\otimes\sigma_{z}]\otimes\mathbb{I}_{8})\ket{\phi_{0}(8)},\\
\ket{b_{2}b_{3}b_{1}}:=\ket{\phi_{3}}=(U_3\otimes\mathbb{I}_{8})\ket{\phi_{0}(8)}=([\sigma_{z}\otimes\sigma_{x}\otimes\mathbb{I}_{2}]\otimes\mathbb{I}_{8})\ket{\phi_{0}(8)},\\
\ket{b_{2}b_{1}b_{3}}:=\ket{\phi_{4}}=(U_4\otimes\mathbb{I}_{8})\ket{\phi_{0}(8)}=([\sigma_{z}\otimes\mathbb{I}_{2}\otimes\sigma_{x}]\otimes\mathbb{I}_{8})\ket{\phi_{0}(8)},\\
\ket{b_{3}b_{1}b_{2}}:=\ket{\phi_{5}}=(U_5\otimes\mathbb{I}_{8})\ket{\phi_{0}(8)}=([\sigma_{x}\otimes\mathbb{I}_{2}\otimes\sigma_{z}]\otimes\mathbb{I}_{8})\ket{\phi_{0}(8)},\\
\ket{b_{3}b_{2}b_{1}}:=\ket{\phi_{6}}=(U_6\otimes\mathbb{I}_{8})\ket{\phi_{0}(8)}=([\sigma_{x}\otimes\sigma_{z}\otimes\mathbb{I}_{2}]\otimes\mathbb{I}_{8})\ket{\phi_{0}(8)}.
\end{align*}
\normalsize
Here $\ket{\phi_{0}(8)}:=\ket{\phi^{+}}^{\otimes 3}\in\mathbb{C}^{8}\otimes\mathbb{C}^{8}$. Now, consider an arbitrary quantum state $\ket{\chi}:=\sum_{i=0}^{7}a_{i}\ket{i}\in \mathbb{C}^{8}$, where $a_{i}\in\mathbb{C}~\&~\sum_{i=0}^{7}|a_{i}|^{2}=1$. Thus the condition for distinguishability of the set $\mathcal{B}_{\mathcal{P}[\{3\}]}$ under one-way LOCC turns out to be $\langle\psi_i|\psi_j\rangle=\delta_{ij}$, where $\ket{\psi_k}:=U_k\ket{\chi}$. It boils down to a numerical exercise to show that the aforesaid condition is not satisfied for any $\ket{\chi}\in\mathbb{C}^8$.
\end{proof}
While Proposition {\ref{prop4} } proves impossibility of LSM of the set $\mathcal{B}_3$ under one-way LOCC, we have the following stronger result if we consider $2$-LSM of the set. 
\begin{corollary}\label{coro3}
Perfect $2$-LSM of the set $\mathcal{B}_3$ is not possible even under two-way LOCC protocol. 
\end{corollary} 
Proof follows from the fact that $6$ pairwise maximally entangled states in $(\mathbb{C}^4)^{\otimes2}$ are not locally distinguishable. Moving on to the question of entanglement assisted discrimination of the set $\mathcal{B}_3$, it has been established that perfect discrimination requires $1$-ebit entanglement \cite{Bandyopadhyay15}. Regarding entanglement assisted marking of $\mathcal{B}_3$ we have the following result. 
\begin{proposition}\label{prop5}
There exists a $(1,1)$ entanglement catalytic protocol for perfect LSM of the set $\mathcal{B}_3$. 
\end{proposition} 
\begin{proof}
Let Alice and Bob have $1$-ebit entanglement (received from some supplier) to distinguish the state $\ket{b_p}\otimes\ket{b_q}\otimes\ket{b_r}$, where $p,q,r\in\{1,2,3\}~ \&~p, q, r$ are distinct. Using the teleportation scheme they can identify one of the indices (say) $p$. Then, using the method of Walgate {\it et al.} \cite{Walgate00}, they identify the remaining two indices. At the end of this protocol $1$-ebit entanglement remains with Alice and Bob which they can return to the supplier. So, in a catalytic sense, the protocol consumes $0$-ebit of entanglement.    
\end{proof}
Note that the protocol in Proposition \ref{prop5} involves two-way CC. If the teleportation step is from Alice to Bob and thus requires CC from Alice to Bob, then the Walgate step requires CC from Bob to Alice. The question remains open whether there exists some local protocol with two-way CC that perfectly marks the set $\mathcal{B}_3$ without involving entanglement even in the catalytic sense.   

\section{ Discussions}
To further highlight the implication of the results from the previous section, a few comments are in order. Although both the problems of LSM and LSD stem from a common notion of state identification, the present work strives to point out a subtle difference between them. To elaborate this difference one can consider the following three-party information theoretic task.

Let us suppose three parties Alice, Bob and Charlie are spatially separated. Charlie shares quantum transmission lines with both Alice and Bob, but Alice and Bob are restricted to classical communication between themselves only. Charlie would like to communicate a classical message to both Alice and Bob. But to do that he is provided with an ensemble of $n$ orthogonal bi-partite states of local dimension $d$ which are not locally distinguishable. A justification for communicating in this way is to avoid the message being decoded by non-communicating eavesdroppers between Charlie-Alice and Charlie-Bob.

Now Charlie can provide Alice and Bob multiple copies of the unknown state from the ensemble, so that they can perform perfect LSD. Let us suppose $k$ copies are necessary for perfect LSD. Thus Charlie could communicate to Alice and Bob $\log~n$ bits by sending $k$ qudits, i.e. $\frac{\log~n}{k}$ bits per qudit.

Alternatively Charlie can provide Alice and Bob states from the ensemble corresponding to LSM task, i.e. an ensemble of size $\log~n!$. Possibility of perfect LSM of this ensemble under LOCC will result in a communication of $\log~n!$ bits by sending $n$ qudits, i.e. $\frac{\log~n!}{n}$ bits per qudit.  

To compare the average communication per qudit, let us consider the ensemble in Theorem \ref{theo2} of the main text. The ensemble $\mathcal{X}_4$ of $4$ orthogonal states with local dimension $4$ is given to Charlie. This ensemble does not allow perfect LSD (according to Theorem \ref{theo2}) but $2$ copies of the unknown state is sufficient for perfect LSD. So the average communication per ququad is $\frac{\log~4}{2}=1$ bit. On the other hand perfect LSM of this ensemble (as in Theorem 2) implies average communication per ququad to be $\frac{\log~4!}{4}=\frac{\log~24}{4}=\frac{3+\log~3}{4}$ bits which is greater than the average communication for the protocol based on multi-copy LSD. In this sense, LSM is more economical over the conventional multi-copy LSD.

Proposition \ref{prop4} is also interesting from a different perspective. It is known that any set of $d+1$ mutually orthogonal $d\otimes d$ maximally entangled states is locally indistinguishable \cite{Hayashi06}. But answer to the same question for smaller sets $(<d+1)$ is known only in a few cases.  Although the result of Walgate {\it et al.} ensures local distinguishability of any two maximally entangled states in $2\otimes2$ and later Nathanson proved that any three mutually orthogonal $3\otimes3$ maximally entangled states are locally distinguishable \cite{Nathanson05}, the authors in \cite{Ghosh04,Yu12,Singal17} provide examples of $4$ maximally entangled states in $4\otimes4$ that are not local distinguishable. In Ref.\cite{Bandyopadhyay11(0)} one can find an example of $4$ maximally entangled states in $5\otimes5$ as well as an example of $5$ maximally entangled states in $6\otimes6$ that cannot be perfectly distinguished under one-way LOCC. In a similar spirit, the set $\mathcal{B}_{\mathcal{P}[\{3\}]}$ constitutes an example of $6$ maximally entangled states in $8\otimes8$ that cannot be distinguished under one-way LOCC. 

\section{Conclusions}
We have proposed a class of novel discrimination tasks, namely the $m$-LSM task, that goes beyond the much explored task of local state discrimination. The present study unravels several curious and intricate features of the proposed task.  

Although Lemma \ref{lemma1}, Corollary \ref{coro1}-\ref{coro2}, and Theorem \ref{theo1}-\ref{theo2} unveil some general  features of local state marking task and  Proposition \ref{prop1}-\ref{prop5} report some interesting consequences by considering specific set of states, the present work leaves open a number of important questions and possibilities for further study. In the following we summarize some of those. First, it is important to resolve the question of optimal resource consumption for local state marking task with and without catalysts as mentioned in the discussions after Theorem \ref{theo2} \& Proposition \ref{prop3}, respectively. Second, all the ensembles considered in the present work consist of bipartite entangled states. Except Corollary \ref{coro2}, the present work does not provide much insight for the local state marking of ensembles containing only product states. Does there exist such a product ensemble that cannot be marked locally? If yes, would it imply a stronger notion of {\it nonlocality without entanglement} \cite{Bennett99}? In recent past, this phenomena of {\it nonlocality without entanglement} has been studied in the generalized probabilistic theory framework \cite{Bhattacharya20}. It might be interesting to extend the study of LSM in this framework. Third, in the same spirit of multipartite LSD \cite{Hayashi06}, exploring LSM for multipartite systems might unveil new features of LOCC as well as of multipartite entanglement. Finally, local indistinguishability has also been shown to have practical implications in cryptographic primitives such as data hiding and secret sharing. It would be quite interesting to find such novel applications for the LSM task introduced here.

{\bf Acknowledgment:} R.K.P. acknowledges support from the CSIR Project No. 09/997(0079)/2020-EMR-I. T.G. was supported by the Hong Kong Research Grant Council through Grant No. 17300918 and through the Senior Research Fellowship Scheme SRFS2021-7S02. S.S.B. acknowledges partial support by the Foundation for Polish Science (IRAP project, ICTQT, Contract No. MAB/2018/5, co-financed by EU within Smart Growth Operational Programme). M.A. and M.B. acknowledge support through the research grant of INSPIRE Faculty fellowship from the Department of Science and Technology, Government of India. M.B. acknowledges funding from the National Mission in Interdisciplinary Cyber-Physical systems from the Department of Science and Technology through the I-HUB Quantum Technology Foundation (Grant No. I-HUB/PDF/2021-22/008) and the start-up research grant from SERB, Department of Science and Technology (Grant No. SRG/2021/000267).
\onecolumngrid
\appendix
\section{Detailed proof of Theorem \ref{theo2}}\label{appendix}
\begin{proof}
Here we show that the set $\mathcal{X}_4\equiv\{\ket{\chi_i}\}_{i=1}^4\subset\mathbb{C}^4_A\otimes\mathbb{C}^4_B$ allows perfect LSM ( more particularly, $4$-LSM); where 
\begin{align*}
\ket{\chi_1}&:=\ket{\phi^+}_{A_1B_1}\otimes\ket{\phi^+}_{A_2B_2},~\ket{\chi_2}:=\ket{\phi^-}_{A_1B_1}\otimes\ket{\phi^-}_{A_2B_2},\\
\ket{\chi_3}&:=\ket{\psi^+}_{A_1B_1}\otimes\ket{\phi^-}_{A_2B_2},~\ket{\chi_4}:=\ket{\psi^-}_{A_1B_1}\otimes\ket{\phi^-}_{A_2B_2},\\
\mbox{with}&~~\ket{\phi^\pm}:=\frac{\ket{00}\pm\ket{11}}{\sqrt{2}},~~~\ket{\psi^\pm}:=\frac{\ket{01}\pm\ket{10}}{\sqrt{2}}.
\end{align*}
The $A_1,A_2$ subsystems are with Alice while $B_1,B_2$ are with Bob. As already mentioned, the part of the state $\ket{\chi_i}$ indexed with $A_1B_1$ will be called the first part and the part with index $A_2B_2$ will be called the second part. The set $\mathcal{X}_4$ can be thought of as the union of two disjoint sets of states 
\begin{align*}
\mathcal{X}_4=G_1\cup G_2~~&\&~~G_1\cap G_2=\emptyset,\\
\mbox{where},~~G_1 := \{\ket{\chi_1},\ket{\chi_2}\},&~~G_2 := \{\ket{\chi_3},\ket{\chi_4}\}. 
\end{align*}
The LSM task can be considered as identifying the indices $p,q,r,s\in\{1,\cdots,4\}$ in the state $\ket{\Sigma}:=\ket{\chi_p}\otimes\ket{\chi_q}\otimes\ket{\chi_r}\otimes\ket{\chi_s}\in\left(\mathbb{C}^4_A\otimes\mathbb{C}^4_B\right)^{\otimes4}$ locally, where $p,q, r, s$ are distinct. Note that the state $\ket{\Sigma}$ is a composition (tensor product) of four different states and we will call $\ket{\chi_p}$ the first state, $\ket{\chi_q}$ the second state and so on (of course, the indices $p,q,\cdots$ are not known and the aim is to identify them locally). The local marking strategy of Alice and Bob goes as follows:
\begin{figure}[t!]
\centering
\begin{picture}(120,440)
\put(-190,0){\includegraphics[width=1.0\textwidth]{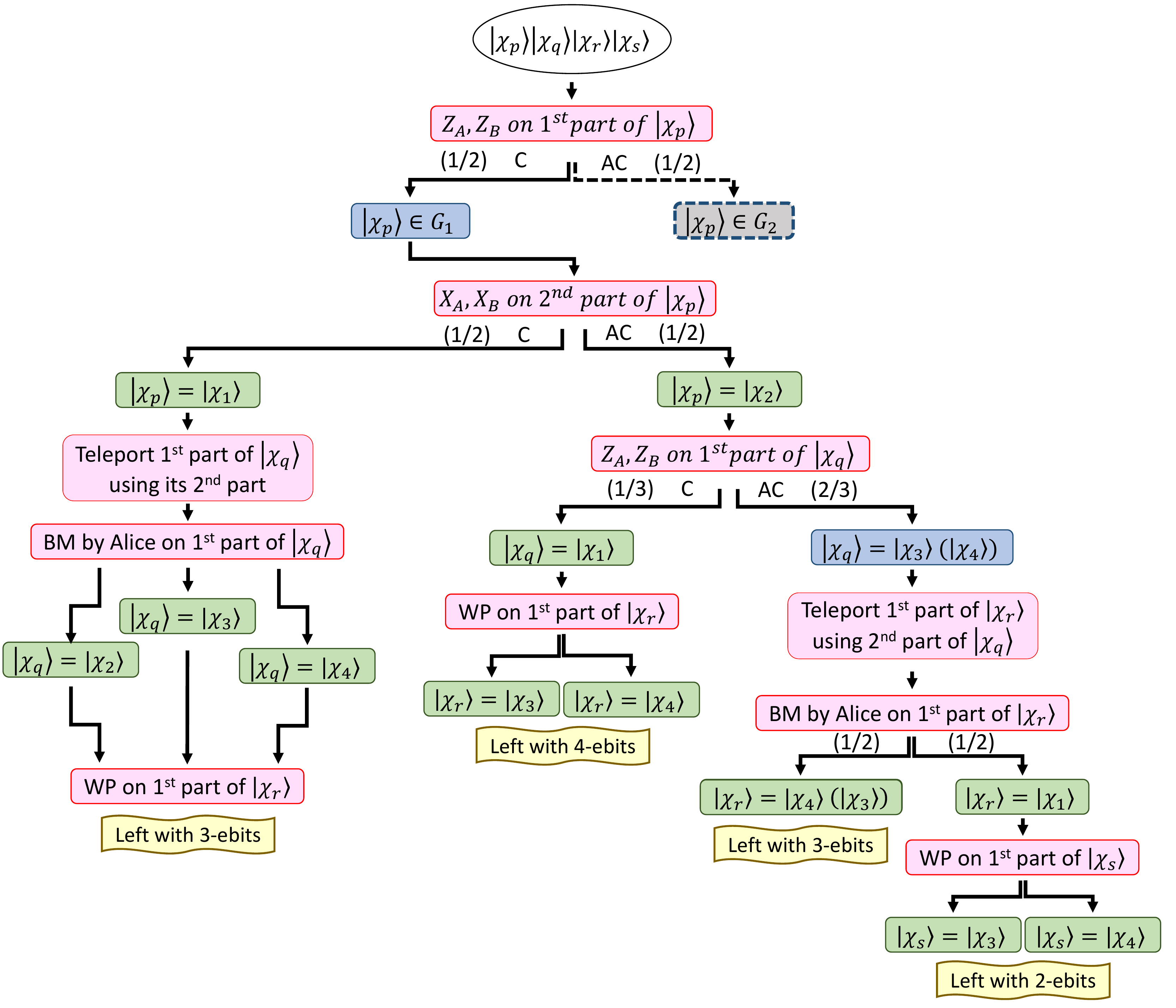}}
\put(105,323){See Figure \ref{fig2}}
\put(-32,382){Step-\ref{Word:step1}}
\put(-32,305){Step-\ref{Word:step2}}
\put(-194,237){Step-\ref{Word:step3c1}}
\put(-199,227){[Case-I]}
\put(197,242){Step-\ref{Word:step3c2}}
\put(193,232){[Case-II]}
\put(-190,100){Step-\ref{Word:step4c1}}
\put(-195,90){[Case-I]}
\put(270,168){Step-\ref{Word:step4c2}}
\put(265,158){[Case-II]}
\end{picture}
\caption{(Color online) Flow chart of the protocol if correlated outcomes are obtained in Step-1. The number written in each branch indicates the probability of occurrence of that branch. The average amount of entanglement left in this case is $\left[\frac{1}{2}\times3+\frac{1}{2}\left\{\frac{1}{3}\times4+\frac{2}{3}\left(\frac{1}{2}\times3+\frac{1}{2}\times2\right)\right\}\right]=3$ ebits.}\label{fig1}
\end{figure}
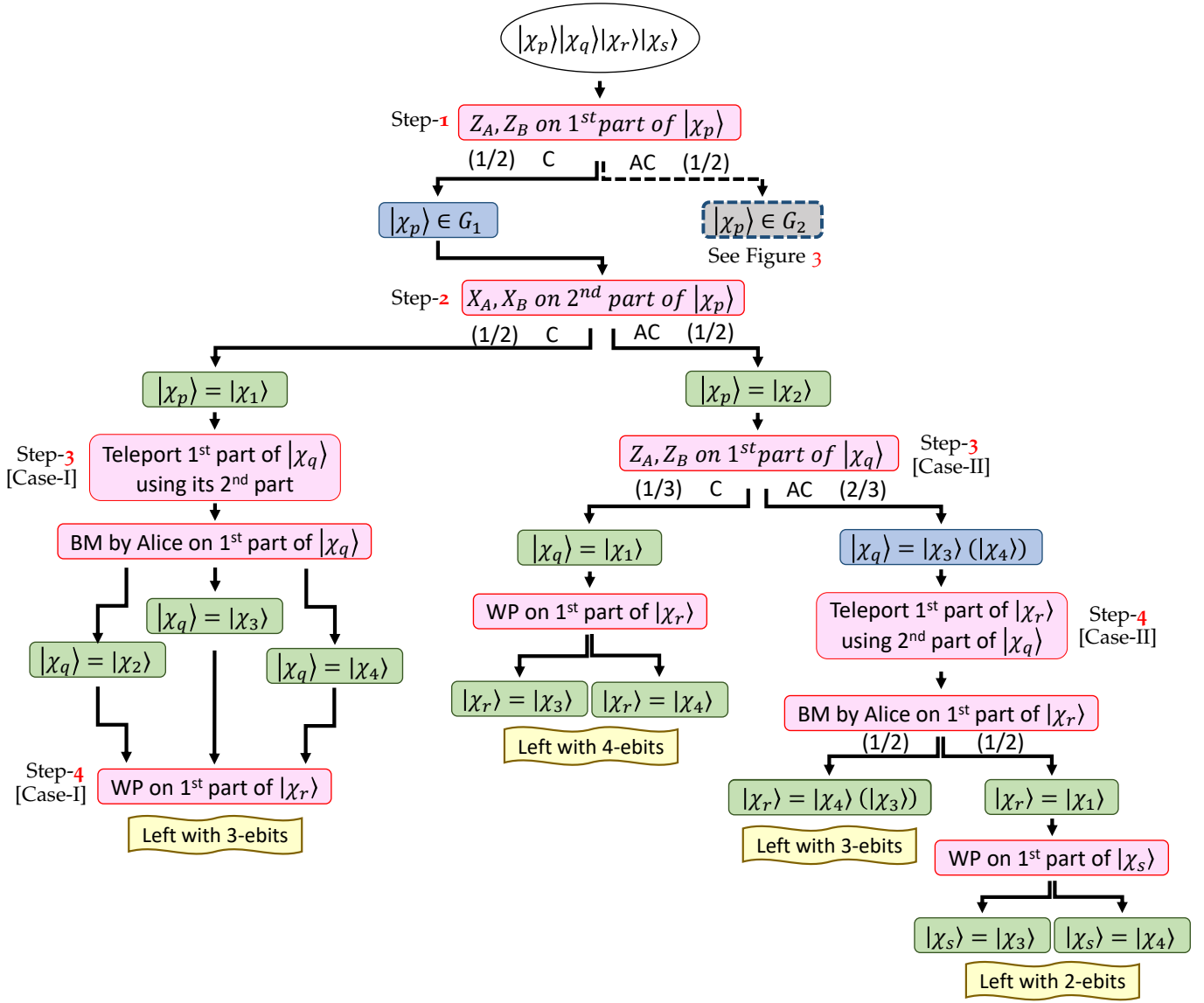
	
{\bf Step-}\setword{{\bf 1}}{Word:step1}{\bf:} Both Alice and Bob perform the Pauli-Z measurement on the first part of the first state ({\it i.e.,} the state $\ket{\chi_p}$). If they obtain correlated (C) outcomes, {\it i.e.,} If Alice and Bob both obtain the same outcome, then they conclude that $\ket{\chi_p}\in G_1$, whereas anti-correlated (AC) outcomes imply $\ket{\chi_p}\in G_2$. Depending on the results obtained in Step-1 they determine their protocol for the next step. For instance, if they obtain correlated outcomes then their protocol is discussed below. 
\begin{figure}[t!]
\centering
\begin{picture}(120,530)
\put(-190,0){\includegraphics[width=1.0\textwidth]{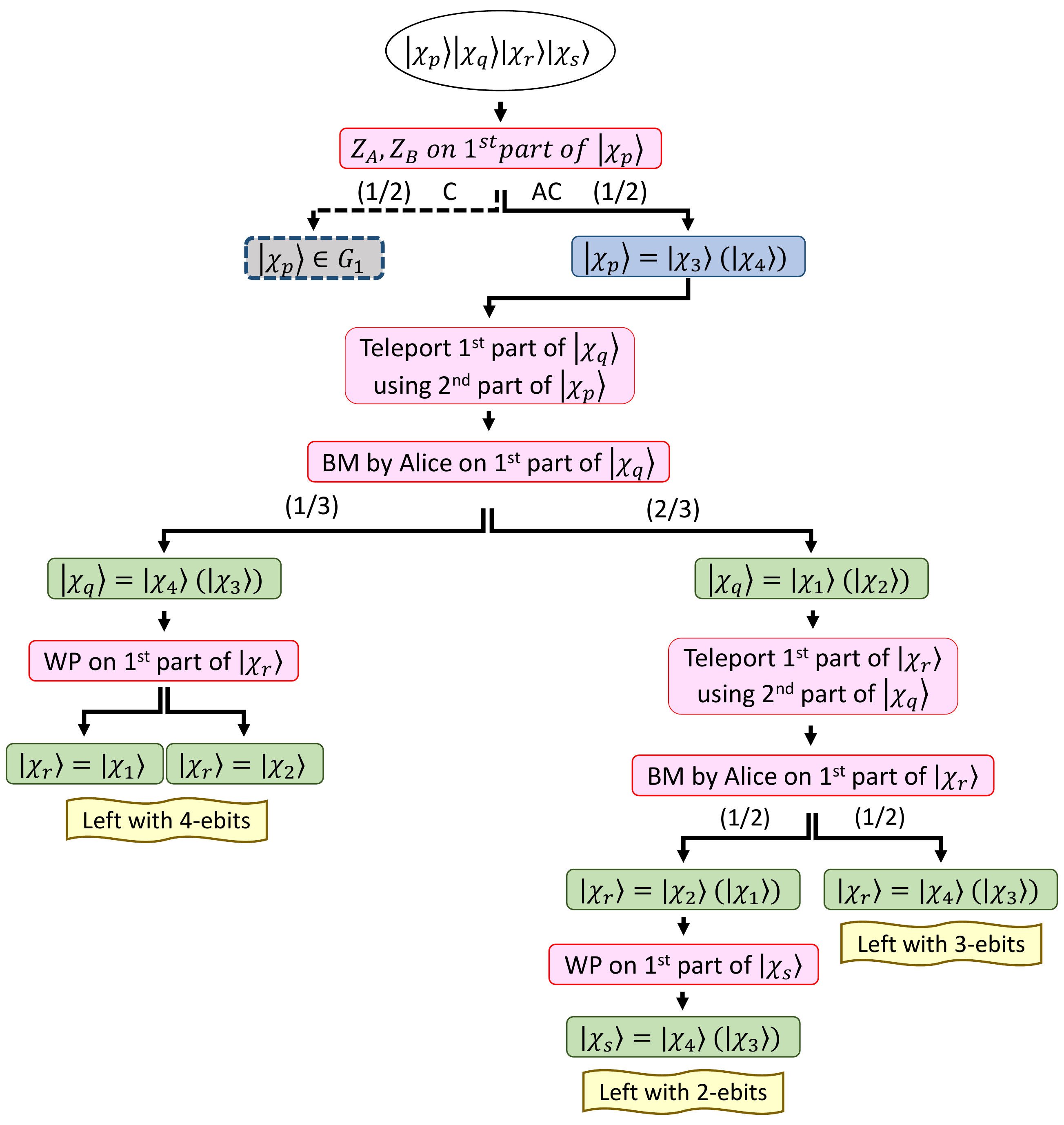}}
\put(-62,393){See Figure \ref{fig1}}
\end{picture}
\caption{(Color online) Flow chart of the protocol if anti-correlated outcomes are obtained in Step-1. The average amount of entanglement left in this case is $\left[\frac{1}{3}\times4+\frac{2}{3}\left\{\frac{1}{2}\times2+\frac{1}{2}\times3\right\}\right]=3$ ebits.}\label{fig2}
\end{figure}
	
{\bf Step-}\setword{{\bf 2}}{Word:step2}{\bf:} Knowing that $\ket{\chi_p}\in G_1$, both Alice and Bob perform Pauli-X measurement on the second part of $\ket{\chi_p}$. Correlated outcomes imply that the first state is $\ket{\chi_1}$ ({\it i.e.,} $p=1$), else it is $\ket{\chi_2}$ ({\it i.e.,} $p=2$). Accordingly, two different branches open up at the next step.  
	
{\bf Step-}\setword{{\bf 3}}{Word:step3c1} {\bf :[Case-I]} $p=1$ in Step-2 implies that the second part of all the states $\ket{\chi_q},\ket{\chi_r}~\&~\ket{\chi_s}$ is $\ket{\phi^-}$. Using the second part of the second state ({\it i.e.,} $\ket{\chi_q}$) Alice and Bob follow the teleportation protocol (TP) to prepare the first part of $\ket{\chi_q}$ at Alice's laboratory. Alice now performs the Bell basis measurement (BM) on the first part of $\ket{\chi_q}$ and depending upon the measurement outcome marks the state exactly. 

{\bf Step-}\setword{{\bf 4}}{Word:step4c1} {\bf :[Case-I]} Since two states $\ket{\chi_p}$ and $\ket{\chi_q}$ are marked exactly (in this case $p=1$ and $q=2~or~3~or~4$), the result of Walgate {\it et al.} \cite{Walgate00} allows us to mark the state $\ket{\chi_r}$ by a local protocol on the first part of the state (In the flow charts of Figure \ref{fig1} \& \ref{fig2} we will call it the Walgate Protocol and denote it as WP). The remaining state $\ket{\chi_s}$ is immediately marked as the set $\mathcal{X}_4$ is known. For the sake of completeness we list the different possibilities and the corresponding Walgate Protocols:
\begin{itemize}
\item $p=1$ (in Step-2) and $q=2$ (in Step-3  [Case- I]): Both Alice and Bob perform the Pauli-X measurement on the first part of $\ket{\chi_r}$. Correlated outcomes imply $r=3$ and $s=4$. Anti-correlated outcomes imply $r=4$ and $s=3$.
\item $p=1$ (in Step-2) and $q=3$ (in Step-3 [Case- I]): Both Alice and Bob perform the Pauli-Z measurement on the first part of $\ket{\chi_r}$. Correlated outcomes imply $r=2$ and $s=4$. Anti-correlated outcomes imply $r=4$ and $s=2$.
\item $p=1$ (in Step-2) and $q=4$ (in Step-3 [Case- I]): Both Alice and Bob perform the Pauli-Z measurement on the first part of $\ket{\chi_r}$. Correlated outcomes imply $r=2$ and $s=3$. Anti-correlated outcomes imply $r=3$ and $s=2$.
\end{itemize}
Note that the entanglement of $\ket{\chi_p}$, $\ket{\chi_q}$, and the first part of $\ket{\chi_r}$ gets destroyed in the protocol, whereas the entanglement of $\ket{\chi_s}$ and the second part of $\ket{\chi_r}$ remains intact. So, whatever the outcome of BM at Step-3, the protocol ends with $3$-ebit entanglement that can be used as a resource.

{\bf Step-}\setword{{\bf 3}}{Word:step3c2} {\bf :[Case-II]} Let Step-2 yield the conclusion that $p=2$. Then, both Alice and Bob perform the Pauli-Z measurement on the first part of the second state ({\it i.e.,} the state $\ket{\chi_q}$).
\begin{itemize}
\item If correlated outcomes are obtained then $q=1$.
\item If anti-correlated outcomes are obtained then $q=3$ or $4$.
\end{itemize}

{\bf Step-}\setword{{\bf 4}}{Word:step4c2} {\bf :[Case-II]}  If correlated outcome is obtained in Step-3 [Case-II], then we have $p=2$ and $q=1$. Again, the result of Walgate {\it et al.} ensures that local marking of $\ket{\chi_r}$ is possible by a local protocol on the the first part of the state and accordingly the remaining state $\ket{\chi_s}$ is also marked. This leaves us with $4$-ebit of  entanglement at the end of the protocol -- $1$-ebit each in $\ket{\chi_q}$ and $\ket{\chi_r}$, and $2$-ebit in $\ket{\chi_s}$.
	
If anti-correlated outcome is obtained in Step-3 [Case-II] then Alice and Bob know that the second part of $\ket{\chi_q}$ is $\ket{\phi^-}$. Utilizing this $\ket{\phi^-}$ they teleport and prepare the first part of $\ket{\chi_r}$ at Alice's laboratory. Alice performs the Bell basis measurement (BM) on the first part of $\ket{\chi_r}$ and marks the state exactly. 
\begin{itemize}
\item If $\ket{\chi_r}$ is identified as $\ket{\chi_4}$ then we have $p=2,q=3, r=4,s=1$. If $\ket{\chi_r}$ is identified as $\ket{\chi_3}$ then we have $p=2,q=4, r=3,s=1$. In both theses cases we are left with $3$-ebit of entanglement.
\item If the state $\ket{\chi_r}$ is identified as $\ket{\chi_1}$, then we have $p=2$ and $r=1$. WP allows us to mark the state $\ket{\chi_s}$ by a local protocol on its first part. Therefore we have  either $p=2,q=4,r=1,s=3$ or $p=2,q=3,r=1,s=4$. Both these cases leave us with $2$-ebit of entanglement.
\end{itemize}
So far, we have discussed the protocol if we obtain correlated outcomes in Step-1. The protocol is summarized in the flow-chart shown in Figure \ref{fig1}. However, to complete the proof we need to analyze the case if anti-correlated outcomes are obtained in Step-1. The corresponding flow-chart is shown in Figure \ref{fig2}. From the flow charts it straightforwardly follows that on an average $3$-ebit $\left[\frac{1}{2}(3+3)\right]$ of entanglement is left at the end of the protocol. 
\end{proof}

\twocolumngrid

\end{document}